\documentclass[12pt]{article}
\usepackage{amsmath}
\usepackage{amssymb}
\usepackage{amsfonts}
\usepackage{latexsym}
\usepackage{color}

\catcode `\@=11 \@addtoreset{equation}{section}

\catcode `\@=12

\newcommand{\be}{\begin{equation}}
\newcommand{\en}{\end{equation}}
\newcommand{\bea}{\begin{eqnarray}}
\newcommand{\ena}{\end{eqnarray}}
\newcommand{\beano}{\begin{eqnarray*}}
\newcommand{\enano}{\end{eqnarray*}}
\newcommand{\bee}{\begin{enumerate}}
\newcommand{\ene}{\end{enumerate}}

\newcommand{\Hil}{{\cal H}}

\newcommand{\Id}{1\!\!1}
\newcommand{\F}{{\cal F}}
\newcommand{\Lc}{{\cal L}}

\newcommand{\C}{{\cal C}}
\newcommand{\E}{{\cal E}}

\newtheorem{thm}{Theorem}

\newtheorem{lemma}[thm]{Lemma}
\newtheorem{prop}[thm]{Proposition}

\newenvironment{proof}{\noindent {\bf Proof:}}{\hfill$\Box$}

\begin{document}

\thispagestyle{empty}

\vspace*{1cm}

\begin{center}
{\Large \bf Some invariant biorthogonal sets with an application to coherent states}   \vspace{2cm}\\

{\large F. Bagarello}
\vspace{3mm}\\
 DEIM -Dipartimento di Energia, ingegneria dell' Informazione e modelli Matematici,
\\ Fac. Ingegneria, Universit\`a di Palermo, I-90128  Palermo, and INFN, Torino, Italy\\
e-mail: fabio.bagarello@unipa.it
\vspace{2mm}\\
{\large S. Triolo}
\vspace{3mm}\\
  DEIM -Dipartimento di Energia, ingegneria dell' Informazione e modelli Matematici,
\\ Fac. Ingegneria, Universit\`a di Palermo, I-90128  Palermo, Italy\\
e-mail: salvatore.triolo@unipa.it
\end{center}

\vspace*{2cm}

\begin{abstract}
\noindent We show how to construct, out of a certain basis invariant under the action of one or more unitary operators, a second biorthogonal
set with  similar properties. In particular, we discuss conditions for this new set to be also a basis of the Hilbert space, and we apply the
procedure  to coherent states. We conclude the paper considering a simple application of our construction to pseudo-hermitian quantum
mechanics.
\end{abstract}

\vspace{2cm}


\vfill

\newpage

\section{Introduction}

In the mathematical and physical literature many examples of complete sets of vectors in a given Hilbert space $\Hil$ are constructed starting
from a single normalized element $\varphi_0\in\Hil$,
 acting on this vector several time with a given set of unitary
operators. For instance, this is exactly what happens for coherent states and for wavelets. In the first case one essentially acts several
times on the vacuum of a bosonic oscillator with a modulation and with a translation operator. In the second example, to produce a complete set
of wavelets, one acts   on a {\em mother wavelet} with powers of  a dilation and of a translation operator. In this last situation the result of
this action can be an orthonormal (o.n.) set of vectors, and this is, in fact, the main output of the so-called multi-resolution analysis, \cite{dau}. On the
other hand, this is forbidden for general reasons for coherent states. In two previous papers, \cite{bt1,bt2}, we have considered the following
problem: given a fixed element of $\Hil$, $\varphi_0$, and a certain set of unitary operators, $A_1,\ldots,A_N$, and defining new vectors
$\varphi_{k_1,\dots,k_N}:=A_1^{k_1}\cdots A_N^{k_N}\varphi_0$, $k_j\in\Bbb{Z}$ for all $j=1,2,\ldots,n$, is it possible to produce, out of
these vectors, a new vector $\hat\varphi_0$ such that the new vectors $\hat\varphi_{k_1,\dots,k_N}:=A_1^{k_1}\cdots A_N^{k_N}\hat\varphi_0$
turn out to be mutually orthogonal? The answer was, in general, positive, and we have proposed an invariant procedure which, however, must be
solved perturbatively. It should be mentioned that, for coherent states, that approach didn't work in all of $\Lc^2(\Bbb R)$, but only in some
suitable Hilbert subspaces of $\Lc^2(\Bbb R)$.

Here we consider a slightly different problem, which is also physically motivated by the recent interest on pseudo-hermitian quantum mechanics
and by the role that biorthogonal sets necessarily have in this context, because of the absence of a self-adjoint hamiltonian describing the dynamics of the system under consideration, \cite{bender,ali}. More explicitly, the problem we address in this paper is the following: is it
possible, out of the linearly independent vectors $\varphi_{k_1,\dots,k_N}$ above, to construct a new vector $\Psi_0\in\Hil$ such that the
vectors $\Psi_{k_1,\dots,k_N}:=A_1^{k_1}\cdots A_N^{k_N}\Psi_0$ are biorthogonal to the original ones? And, do these vectors define a basis in
$\Hil$, at least under suitable conditions? This is not a trivial question. In fact, it is known that two biorthogonal sets are not necessarily bases, \cite{bases}.

The paper is organized as follows:

in the next section we state the general problem, discuss the method and show some {\em prototype} examples, in $N=1$.

In Section III we discuss in many details the case of the coherent states (N=2, assuming that $A_1A_2=A_2A_1$), and we find conditions for our procedure to work.

In Section IV we briefly discuss how to extend our procedure to $N\geq3$, assuming again that the operators $A_j$ mutually commute. Also, we consider some relations between our construction and pseudo-hermitian quantum mechanics. Our final considerations are given in Section V.

\section{Stating the problem and first results}

Let $\Hil$ be a Hilbert space, $\varphi\in\Hil$ a fixed element of the space and let $A_1,\ldots,A_N$ be $N$ given unitary operators:
$A_j^{-1}=A_j^\dagger$, $j=1,2,\ldots,N$. Let $\Hil_N$ be the closure of the linear span of the set
\be\F_\varphi=\{\varphi_{k_1,\ldots,k_N}:=A_1^{k_1}\cdots A_N^{k_N}\varphi,\,\, k_1,\ldots k_N\in\Bbb{Z} \}.\label{II1}\en Of course, in order
for this situation to be of some interest, we assume that an infinite elements of $\F_\varphi$ are linearly independent, so to have
$dim(\Hil_N)=\infty$. To simplify the treatment, in the following, we will assume that all the vectors $\varphi_{k_1,\ldots,k_N}$ are
independent.  In this case, by construction, $\F_\varphi$ is a basis for $\Hil_N$. However, in general, there is no reason why the vectors in
$\F_\varphi$ should be mutually orthogonal. On the contrary, without a rather clever choice of both $\varphi$ and $A_1,\ldots,A_N$, it is very
unlikely to obtain an o.n. set. As stated in the introduction, our aim is to discuss some general technique which produces, for suitable $A_j$'s, another vector
$\Psi\in\Hil_N$ such that the set \be\F_\Psi=\{\Psi_{k_1,\ldots,k_N}:=A_1^{k_1}\cdots A_N^{k_N}\Psi,\,\, k_1,\ldots k_N\in\Bbb{Z}
\}\label{II2}\en is biorthogonal to $\F_\varphi$, i.e. $\left<\varphi_{k_1,\ldots,k_N},\Psi_{l_1,\ldots,l_N}\right>= \delta_{k_1,l_1}\cdots \delta_{k_N,l_N}$. Moreover, we would like this set to share as much of the  original features of $\F_\varphi$
as possible. For instance, if the set $\F_\varphi$ is a set of coherent states, we would like the new vectors $\Psi_{k_1,\ldots,k_N}$ to be
also coherent states, in some sense, other than being stable under the action of certain unitary operators. But, first of all, we need
$\F_\Psi$ to be a basis for $\Hil_N$.

We will consider our problem step by step, starting with the simplest situation which is, clearly, $N=1$. In this case the set $\F_\varphi$ in
(\ref{II1}) reduces to $\F_\varphi =\{\varphi_{k}:=A^{k}\,\varphi,\,\, k\in\Bbb{Z} \}$ with $<\varphi_k,\varphi_l>\neq \delta_{k,l}$ (otherwise
we can easily solve the problem by taking $\Psi=\varphi$). Since $\F_\varphi$ is a basis for $\Hil_1$, any element in $\Hil_1$ can be written
in terms of the vectors of $\F_\varphi$. Let $\Psi\in\Hil_1$ be the following linear combination:
\be\Psi=\sum_{k\in\Bbb{Z}}\,c_k\varphi_k,\label{II3}\en and let us define more vectors of $\Hil_1$ as
\be\Psi_n:=A^n\Psi=\sum_{k\in\Bbb{Z}}\,c_k\,\varphi_{k+n}=X\varphi_{n},\label{II4}\en where we have introduced the  operator \be
X:=\sum_{k\in\Bbb{Z}}\,c_k\,A^k.\label{II5}\en Then we introduce the set $\F_\Psi=\{\Psi_{k}:=A^{k}\,\Psi,\,\, k\in\Bbb{Z} \}$. Our main effort
will be to find the coefficients $c_k$ in such a way, first of all, the bi-orthogonalization requirement
$<\Psi_n,\varphi_k>=\delta_{n,k}$, $n$, $k$ in $\Bbb Z$, holds. We also want $\F_\Psi$ to be a basis for $\Hil_1$.

 Of course all the expansions above are, at this stage, only
formal. What makes our formulas rigorous is the asymptotic behavior of the coefficients of the expansion $c_n$. We will come back on this point
in a moment.

The first useful result, which follows directly from the previous definitions, is  that $<\Psi_n,\varphi_0>=\delta_{n,0}$ for all
$n\in\Bbb{Z}$, if and only if $<\Psi_n,\varphi_k>=\delta_{n,k}$, $\forall n,k\in\Bbb{Z}$. For this reason, in order to fix the coefficients
$c_n$, it is enough to require the bi-orthogonality condition $<\Psi_n,\varphi_0>=\delta_{n,0}$, which becomes \be
\delta_{n,0}=<\Psi_n,\varphi_0>=\sum_{k\in\Bbb{Z}}\,\overline{c_k}\,<\varphi_{k+n},\varphi_0>=\sum_{k\in\Bbb{Z}}\,\overline{c_k}\,\alpha_{k+n},
\label{II6}\en where we have defined \be \alpha_j=<A^j\varphi_0,\varphi_0>,\label{II7}\en $j\in\Bbb Z$. If we now multiply both sides of
(\ref{II6}) for $e^{ipn}$, $p\in[0,2\pi[$, and we sum up on $n\in\Bbb{Z}$, we get \be \overline{C(p)}\,\alpha(p)=1,\quad \mbox{a.e. in }
[0,2\pi[.\label{II8}\en Here we have introduced the following functions:\be C(p)=\sum_{l\in\Bbb{Z}}\,c_l\,e^{ipl},\quad
\alpha(p)=\sum_{l\in\Bbb{Z}}\,\alpha_l\,e^{ipl}.\label{II9} \en Again, these series are not necessarily convergent, so that they must be
considered only as formal objects at this stage. However, we can move from formal to {\em true} functions using the following result:

\begin{lemma}\label{lemma1} $\alpha(p)\in \Lc^2(0,2\pi)$ if and only if $\{\alpha_l\}\in l^2(\Bbb Z)$. In this case
$\|\alpha\|^2=2\pi\sum_{l\in\Bbb Z}|\alpha_l|^2$. Analogously, $C(p)\in \Lc^2(0,2\pi)$ if and only if $\{c_l\}\in l^2(\Bbb Z)$. Then
$\|C\|^2=2\pi\sum_{l\in\Bbb Z}|c_l|^2$. When they both exist finite,  $C(p)$ and $\alpha(p)$ are $2\pi$-periodic and real functions.
\end{lemma}

\begin{proof}

We just prove here the reality of the functions, since the other claims are trivial. The starting point is the following equality:
$\overline{\alpha_l}=\alpha_{-l}$, for all integers $l$, which is a consequence of the definition in (\ref{II7}). Therefore
$\overline{\alpha(p)}=\sum_{l\in\Bbb{Z}}\,\overline{\alpha_l}\,e^{-ipl}=\sum_{k\in\Bbb{Z}}\,\alpha_k\,e^{ipk}=\alpha(p)$, a.e. in $[0,2\pi[$.
>From equality (\ref{II8}) it follows that also $C(p)$ is real.

\end{proof}

To deduce now the expression of $\Psi$, and of the various $\Psi_n$ as a consequence, we need to compute, see (\ref{II3}), the coefficients
$c_l$. These can be deduced by inverting the definition of $C(p)$ in (\ref{II9}) and by using equation (\ref{II8}): \be
c_l=\frac{1}{2\pi}\int_0^{2\pi}\,\frac{e^{-ipl}}{\alpha(p)}\,dp.\label{II10}\en A problem could occur if $\alpha(p)$ has a zero in $[0,2\pi[$.
In this case, in fact, the integral above could be not converging. On the other hand, if $\alpha(p)$ is different from zero a.e. in $[0,2\pi[$,
then all the coefficients $c_l$ surely exist. However, we could  still be in trouble because the series defining $\Psi$ in (\ref{II3}) and $X$
in (\ref{II5}) could be not converging. This depends on how fast the sequence $\{c_l\}$ goes to zero when $l$ diverges. Also, there is no
reason, a priory, for which the properties of $\F_\varphi$ should be shared also by $\F_\Psi$. In particular, there is no reason a priori for
$\F_\Psi$ to be a basis. To address these problems we start noticing that, since $\overline{c_n}=c_{-n}$, it is first clear that $X=X^\dagger$,
if the series for $X$ converges. Now, let us define \be d_l=\frac{1}{2\pi}\int_0^{2\pi}\,e^{-ipl}\,\alpha(p)\,dp, \qquad
Y=\sum_{l\in\Bbb{Z}}d_l\,A^l. \label{IIa}\en Then we have
\begin{prop}\label{prop1}
Let us assume that $\{c_l\}, \{d_l\}\in l^1(\Bbb{Z})$. Then $X$ and $Y$ are bounded operators, and $Y=X^{-1}$. Moreover, $\F_\Psi$ is  the
(only) basis of $\Hil_1$ biorthogonal to $\F_\varphi$.
\end{prop}

\begin{proof}
$X$ and $Y$ are clearly bounded. To prove that $Y=X^{-1}$, we start showing that $\sum_{n\in\Bbb{Z}}c_nd_{l-n}=\delta_{l,0}$. For that we will
use the Poisson summation rule, \cite{dau},
$\sum_{n\in{\Bbb{Z}}}e^{ixan}=\frac{2\pi}{|a|}\sum_{n\in{\Bbb{Z}}}\delta\left(x-\frac{2\pi}{a}n\right)$, $a\neq 0$. Using the definitions of
$c_n$ and $d_n$ we have:
$$
\sum_{n\in\Bbb{Z}}c_nd_{l-n}=\frac{1}{(2\pi)^2}\int_0^{2\pi}dp\int_0^{2\pi}dq\,\frac{\alpha(q)}{\alpha(p)}\,e^{-iql}\sum_{n\in\Bbb{Z}}e^{in(q-p)}.
$$
Since $q, p\in [0,2\pi[$, we have
$$
\sum_{n\in\Bbb{Z}}e^{in(q-p)}=2\pi \sum_{n\in\Bbb{Z}}\delta(q-p-2\pi n)=2\pi\delta(q-p),
$$
so that \be
\sum_{n\in\Bbb{Z}}c_nd_{l-n}=\frac{1}{2\pi}\int_0^{2\pi}dp\int_0^{2\pi}dq\,\frac{\alpha(q)}{\alpha(p)}\,e^{-iql}\delta(q-p)=\frac{1}{2\pi}\int_0^{2\pi}dp\,e^{-ipl}
=\delta_{l,0}. \label{IIb}\en As a consequence,
$$
XY=\sum_{n,m\in\Bbb{Z}}c_n\,d_m\,A^{n+m}=\sum_{l\in\Bbb{Z}}\left(\sum_{n\in\Bbb{Z}}c_n\,d_{l-n}\right)A^{l}=\sum_{l\in\Bbb{Z}}\delta_{l,0}A^l=\Id.
$$
Analogously we find that $YX=\Id$, so that $Y=X^{-1}$.

Therefore, since $X$ and $X^{-1}$ are bounded, and since $\F_\Psi$ is the image via $X$ of a basis of $\Hil_1$, $\F_\varphi$, see (\ref{II4}), $\F_\Psi$ is also
a basis for $\Hil_1$. Biorthogonality is clear while uniqueness follows from \cite{bases}.

\end{proof}

The next Lemma gives sufficient conditions for Proposition \ref{prop1} to be applicable.

\begin{lemma} \label{lemma2} Let $\alpha(p)$ be, together with its derivative $\alpha'(p)$, a function in $\Lc^2(0,2\pi)$. Then $\{d_l\}\in l^1(\Bbb Z)$.
 Moreover, if
$\alpha_{inv}(p):=\frac{1}{\alpha(p)}$ is, together with its derivative $\alpha_{inv}'(p)$, a function in $\Lc^2(0,2\pi)$, then $\{c_l\}\in
l^1(\Bbb Z)$.
\end{lemma}
\begin{proof}
Using the periodicity of $\alpha(p)$ and integrating by parts we get
$$
d_l=\frac{1}{2\pi}\int_0^{2\pi}\,e^{-ipl}\,\alpha(p)\,dp=\frac{1}{2i\pi l}\int_0^{2\pi}\,e^{-ipl}\,\alpha'(p)\,dp.
$$
But, since $\alpha'(p)\in \Lc^2(0,2\pi)$, the modulus of this last integral must decreases to zero at least as $\frac{1}{l^{1/2+\epsilon}}$,
for some positive $\epsilon$. This means that $|d_l|$ decreases to zero at least as $\frac{1}{l^{3/2+\epsilon}}$, and therefore $\{d_l\}\in
l^1(\Bbb Z)$.

The other claim can be proved in the same way.

\end{proof}

\vspace{2mm}

{\bf Remarks:} (1) Equation (\ref{IIb}) shows, in particular, that the sequences $\{c_l\}$ and $\{d_l\}$ might be considered one as a sort of
{\em inverse} of the other. In fact: $\sum_{n\in{\Bbb{Z}}}{c_n}\,d_n=1$.

\vspace{2mm}

(2) Under the assumptions of Proposition \ref{prop1},  if $X$ is also  positive, then the set $\E=\{e_n:=X^{1/2}\varphi_n, \,n\in\Bbb{Z}\}$ is
an orthonormal basis for $\Hil_1$ and $\F_\varphi$ and $\F_\Psi$ are biorthogonal Riesz bases for $\Hil_1$.

\subsection{Simple examples}

We consider now some easy examples of our strategy, for which most of the computations can be performed analytically.

\subsubsection{Example 1}

 Let $\varphi_0(x)=\chi_{[0,a[}(x)$ be the characteristic function of the interval $[0,a[$, with $a>0$, and let $A$ be the
following translation operator: $A=e^{-i\hat p}$. We have
$$\F_\varphi=\{\varphi_{n}(x):=A^{n}\varphi_0(x)=\chi_{[n,n+a[}(x),\,\,
n\in\Bbb{Z} \}.$$ Let $\Hil_1$ be the Hilbert space spanned by these functions. Clearly $\Hil_1\subset\Lc^2(\Bbb R)$. We want to see what our
procedure produces starting with this set. For that, it is convenient to consider separately the cases $a<1$, $a=1$ and $a>1$. Let us start
with the easiest case, $a=1$. In this case the set $\F_\varphi$ is  made by o.n.  functions. Indeed we have
$\alpha_j=<\varphi_j,\varphi_0>=\delta_{j,0}$. Therefore $\alpha(p)=1$, which is obviously never zero, $2\pi$-periodic and square integrable in
$[0,2\pi[$. Moreover, $c_l=d_l=\delta_{l,0}$, see (\ref{II10}) and (\ref{IIa}). From (\ref{II4}) we deduce that $\Psi_n(x)=\varphi_n(x)$ for
all integer $n$. It is clear that both $X$ and $Y=X^{-1}$ exist, and they  both coincide with the identity operator.

Just a little less trivial is the situation when $a<1$. In this case, in fact, the set $\F_\varphi$ is still made of orthogonal functions,
since each $\varphi_{n}(x)=\chi_{[n,n+a[}(x)$ does not overlap with any other $\varphi_{k}(x)=\chi_{[k,k+a[}(x)$, if $k\neq n$. However, none
of these functions is normalized so that we may expect that our procedure simply {\em cures} this feature. Indeed we have
$\alpha_j=\left<\varphi_j,\varphi_0\right>=a\delta_{j,0}$, so that $\alpha(p)=a$, which is again never zero, $2\pi$-periodic and square
integrable in $[0,2\pi[$. We deduce $c_l=\frac{1}{a}\,\delta_{l,0}$ and $ d_l=a\,\delta_{l,0} $ . Therefore, among other properties, $\{c_l\},
\{d_l\}\in l^1(\Bbb{Z})$, so that Proposition \ref{prop1} applies. Indeed, $X=\frac{1}{a}\Id$ is bounded with bounded inverse,
$\Psi_n(x)=\frac{1}{a}\,\varphi_n(x)$,
 $\forall\,n\in \Bbb Z$, and $\left<\Psi_n,\varphi_k\right>=\delta_{n,k}$: the set $\F_\Psi$ is a basis of $\Hil_1$ which is biorthogonal to $\F_\varphi$.

Surely more interesting is the case $a>1$. We restrict ourselves, for the time being, to $1<a<2$. The overlap coefficients $\alpha_j$ can
  be written as
$\alpha_j=a\,\delta_{j,0}+(a-1)\left(\delta_{j,-1}+\delta_{j,1}\right)$, so that $\alpha(p)=a+2(a-1)\cos(p)$. This is a nonnegative, real and
$2\pi$-periodic function, as expected, and furthermore it is never zero in $[0,2\pi[$ since it has a minimum in $p=\pi$ and
$\alpha(\pi)=2-a>0$. If we fix, just to be concrete, $a=\frac{3}{2}$, we can compute analytically
$\sum_{l\in{\Bbb{Z}}}\,|c_l|^2=\frac{1}{2\pi}\int_0^{2\pi}\frac{dp}{\alpha(p)}=\frac{12}{5\sqrt{5}}$. Therefore the sequence $\{c_l\}$ belongs
to $l^2({\Bbb{Z}}).$ Moreover  Lemma \ref{lemma2} guarantees that the sequences $\{c_l\}$ and $\{d_l\}$ are also in $l^1({\Bbb{Z}}).$ For
instance, an easy computation shows that $$ d_l=\frac{1}{2 \pi}\int_0^{2 \pi} e^{-ipl} (\cos(p)+3/2) dp  = \frac{3}{2}\delta_{l,0}+ \frac{1}{2}
(\delta_{l,-1} + \delta_{l,1} ).
 $$
A bit more complicated, but still analytically doable, is the computation of $c_l$. Incidentally, is it possible to check that
 $\Sigma_{n\in{\Bbb{Z}}} c_n d_{n-l}=\delta_{l,0} $.

  For $a\geq 2$ the function $\alpha(p)$ admits a zero.
For instance, if $a=2$ the overlap coefficients are the same as for $a\in]1,2[$,
$\alpha_j=a\,\delta_{j,0}+(a-1)\left(\delta_{j,-1}+\delta_{j,1}\right)=2 \,\delta_{j,0}+\left(\delta_{j,-1}+\delta_{j,1}\right)$, so that
$\alpha(p)=2+2\cos(p)$. This  is zero for $p=\pi$ and one can check that $\sum_{l\in{\Bbb{Z}}}\,|c_l|=\sum_{l\in{\Bbb{Z}}}\,|c_l|^2=+\infty$:
in this case, our framework does not work.

\subsubsection{Example 2}

Another interesting and easy example can be constructed as follows: let $\varphi_0(x)=\chi_{[0,1[}(x)$ and let $A$ be the dilatation operator
$(Ah)(x)=\sqrt{2}\,h(2x)$, $\forall\,h(x)\in\Lc^2(\Bbb{R})$. Then the set $\F_\varphi$ turns out to be
$$
\F_\varphi=\left\{\varphi_n(x)=2^{n/2}\varphi(2^nx)=2^{n/2}\left\{
\begin{array}{ll}
1, \quad\mbox{if }0< x\leq 2^{-n},  \\
0, \quad\mbox{otherwise},
\end{array}
\right.\quad n\in\Bbb{Z}\right\}.
$$
In this case all the overlap coefficients $\alpha_j$ are different from zero. Indeed we get $\alpha_j=2^{-|j|/2}$, for all $j\in\Bbb{Z}$. Since
$\left|\frac{e^{\pm ip}}{\sqrt{2}}\right|=\frac{1}{\sqrt{2}}<1$, it is easy to compute the analytic expression of $\alpha(p)$ and it turns out
that $\alpha(p)=\frac{1}{3-2^{3/2}\cos(p)}$. The minimum of $\alpha(p)$ is found again for $p=\pi$, and $\alpha(\pi)=\frac{1}{3+2^{3/2}}\simeq
0.1716$, which is different from zero. Moreover we find that $\max(\alpha(p))=\alpha(0)=\frac{1}{3-2^{3/2}}\simeq 5.8284$. The $\|.\|_2$-norm
of the sequence $\{c_l\}$ can be computed analytically and we find
$\sum_{l\in{\Bbb{Z}}}\,|c_l|^2=\frac{1}{2\pi}\int_0^{2\pi}\frac{dp}{\alpha(p)}=3$.
 Using  Lemma \ref{lemma2} we conclude that the sequences $\{c_l\}$ and $\{d_l\}$ are in $l^1({\Bbb{Z}}).$
In particular, for instance, the coefficients $c_l$ look like  \be \label{clesempio2} c_l=\frac{1}{2\pi}\int_0^{2\pi}\, e^{-ipl}
(3-{2}^{\frac{3}{2}} \cos(p)) dp=3\delta_{l,0}-\sqrt{2}(\delta_{l,1}+\delta_{l,-1}), \en so that $X=3 \Id -\sqrt{2}\left( A+ A^{-1}\right)$.


 Of course, we could use these coefficients to define the {\em new} set of vectors which are biorthogonal to $\F_\varphi$ using (\ref{II3}) and (\ref{II4}).

\subsubsection{A possible generalization}

It is not difficult to generalize the previous example. For that, we suppose that the operator $A$ and the seed vector $\varphi_0$ satisfy the
following condition
$$\alpha_j=<A^j\varphi_0, \varphi_0>=r^{\mid j \mid},$$
where $r$ is a fixed real quantity with $0\leq r<1$. In particular, if $r=\frac{1}{\sqrt{2}}$, we go back to Example 2. Similar examples could
be constructed by replacing the original dilation, typical of a {\em standard} multi-resolution analysis, \cite{dau}, with a different
dilation, for instance with $(\tilde Ah)(x)=\sqrt{3}h(3x)$. Then the series for $\alpha(p)$ is convergent and we get
$$
\alpha(p) =\sum_{l\in \Bbb Z}\alpha_le^{ipl}=\frac{1-r^2}{1+r^2-2r\cos(p)}.
$$

Lemma \ref{lemma2} can be applied, and the sequences $\{c_l\}$ and $\{d_l\}$ are both in $l^1({\Bbb{Z}})$. In particular, for instance,
$\{c_l\}$ is a finite sequence

  \be \label{pre} c_l=\frac{1}{2\pi}\int_0^{2\pi}\,
\frac{e^{-ipl} ({1+r^2-2r\cos{p}})}{{1-r^2}}dp=\delta_{l,0}\frac{1+r^2}{1-r^2}-(\delta_{l,1}+\delta_{l,-1})\frac{r}{1-r^2}.\en

Therefore $\Psi_0$ is simply a linear combination of $\varphi_0$, $\varphi_1$ and $\varphi_{-1}$.

\vspace{3mm}

The extension of this results to $N\geq2$ is straightforward whenever the operators $A_j$ in (\ref{II1}) mutually commute, which is the relevant case, for instance, for the coherent states we are going to consider, in many details, in Section III. At the end of Section III we will also briefly sketch what happens for $N\geq3$.

\section{Coherent states}

This section is devoted to a  more interesting example involving coherent states, \cite{aag,gazbook}. We will see that the set of coherent
states fits (and extend to $N=2$) the general discussion of Section II, and we will show how and when the bi-orthogonalization procedure works.

Let $\hat q$ and $\hat p$ be the position and momentum operators on a Hilbert space $\Hil$, $[\hat q,\hat p]=i\Id$, and let us now introduce
the following unitary operators: \be U(\underline n)=e^{ia(n_1\hat q-n_2\hat p)}, \quad D(\underline n)=e^{z_{\underline{n}}
b^\dagger-\overline{z}_{\underline n} b},\quad T_1:=e^{ia\hat q},\quad T_2:=e^{-ia\hat p}.\label{31}\en Here $a$ is a real constant satisfying
$a^2=2\pi L$ for some $L\in\Bbb{N}$, while $z_{\underline{n}}$ and $b$ are related to ${\underline{n}}=(n_1,n_2)$ and $\hat q$, $\hat p$ via
the following equalities: \be z_{\underline{n}}=\frac{a}{\sqrt{2}}(n_2+in_1),\quad b=\frac{1}{\sqrt{2}}(\hat q+i\hat p). \label{32a}\en With
these definitions it is clear that \be U(\underline n)=D(\underline
n)=(-1)^{Ln_1n_2}T_1^{n_1}T_2^{n_2}=(-1)^{Ln_1n_2}T_2^{n_2}T_1^{n_1},\label{33}\en where we have also used the commutation rule $[T_1,T_2]=0$,
which follows from the possible values of $a$.

Let $\varphi_{\underline{0}}$ be the vacuum of $b$, $b\,\varphi_{\underline{0}}=0$, and let us define the following {\em coherent states}: \be
\varphi_{\underline{n}}^{(L)}:=T_1^{n_1}T_2^{n_2}\varphi_{\underline{0}}=T_2^{n_2}T_1^{n_1} \varphi_{\underline{0}} =(-1)^{Ln_1n_2}U(\underline
n)\varphi_{\underline{0}}=(-1)^{Ln_1n_2}D(\underline n)\varphi_{\underline{0}}.\label{34}\en Notice that $\varphi_{\underline{n}}^{(L)}$ is exactly as in (\ref{II1}), with $N=2$, identifying $T_j$ with $A_j$. Notice also that, as it is needed, $[T_1,T_2]=0$.

Here and in the following we will write explicitly
the label $L$ whenever this will be important for us, but not everywhere. For instance, $z_{\underline{n}}$ also depends on $L$, but we will
not stress this dependence. On the other hand, since $b$ does not depend on $L$, $\varphi_{\underline{0}}$ is also independent of $L$. However,
since the unitary operators $T_1$ and $T_2$ do depend on $a$, and on $L$ as a consequence, $T_1^{n_1}T_2^{n_2}\varphi_{\underline{0}}$ are also
$L$-dependent. It is very well known that the set of these vectors,
$\F_\varphi^{(L)}=\{\varphi_{\underline{n}}^{(L)},\,\underline{n}\in{\Bbb{Z}}^2\}$, satisfies, among the others, the following properties:
\begin{enumerate} \label{propcoerenti}
\item $\F_\varphi^{(L)}$ is invariant under the action of
$T_j^{n_j}$, $j=1,2$;
\item each $\varphi_{\underline{n}}^{(L)}$ is an eigenstate of $b$:
$b\,\varphi_{\underline{n}}^{(L)}=z_{\underline{n}}\,\varphi_{\underline{n}}^{(L)}$;\item if $L=1$ they satisfy the {\em resolution of the
identity} $\sum_{\underline{n}\in{\Bbb{Z}}^2}\,|\varphi_{\underline{n}}^{(1)}><\varphi_{\underline{n}}^{(1)}|=\Id_1$, where $\Id_1$ is the
identity in $\Lc^2(\Bbb R)$.
\end{enumerate}
Moreover, it is also well known  that they are not mutually orthogonal. Indeed we have: \be
I_{\underline{n}}^{(L)}:=<\varphi_{\underline{n}}^{(L)},\varphi_{\underline{0}}>=(-1)^{Ln_1n_2}\,e^{-\frac{\pi}{2}L(n_1^2
+n_2^2)}.\label{35}\en Of course, for large $L$  the set $\F_\varphi^{(L)}$ can be considered as {\em approximately orthogonal}, since
$I_{\underline{n}}^{(L)}\simeq 0$ for all $\underline{n}\neq \underline 0$. On the contrary, for small $L$,  the overlap between nearest
neighboring vectors is significantly different from zero.

Our aim is to construct a family of vectors $\F_\Psi^{(L)}$ which shares with $\F_\varphi^{(L)}$ most of the above features and which,
moreover, is biorthogonal to $\F_\varphi^{(L)}$. It is important to recall, see \cite{bt1} and references therein, that the set
$\F_\varphi^{(L)}$ is complete in $\Hil:=\Lc^2(\Bbb R)$ if and only if $L=1$. However, this does not prevent us to define, for each $L\geq 1$,
the following set: \be h_L:=\mbox{linear span}\overline{\left\{\varphi_{\underline{n}}^{(L)},\,\underline{n}\in{\Bbb{Z}}^{ 2}\right\}}^{\|.\|}.
\label{36}\en It is clear that $h_1=\Hil$, while, for $L>1$, $h_L\subset \Hil$. It is also clear that $h_L$ is an Hilbert space for each $L$,
since it is a closed subspace of $\Hil$.

We start our procedure here by extending formula (\ref{II4}) to the present settings: let $\Psi_{\underline{0}}^{(L)}\in h_L$ be the following
linear combination:
\be\Psi_{\underline{0}}^{(L)}=\sum_{\underline{k}\in\Bbb{Z}^2}\,c_{\underline{k}}^{(L)}\varphi_{\underline{k}}^{(L)},\label{31b}\en and let us
introduce more vectors of $h_L$ as
\be\Psi_{\underline{n}}^{(L)}=T_1^{n_1}T_2^{n_2}\Psi_{\underline{0}}^{(L)}=\sum_{\underline{k}\in\Bbb{Z}^2}\,c_{\underline{k}}^{(L)}\,
\varphi_{\underline{k}+\underline{n}}^{(L)}=X^{(L)}\varphi_{\underline{n}}^{(L)},\label{32}\en where we have introduced the  operator \be
X^{(L)}=\sum_{\underline{k}\in\Bbb{Z}^2}\,c_{\underline{k}}^{(L)}\,T_1^{k_1}T_2^{k_2}.\label{33b}\en Of course, since $c_{\underline{n}}^{(L)}$
also depends on $L$, a similar dependence is also shared by $X^{(L)}$, and we are making it explicit.  Then we introduce the set
$\F_\Psi^{(L)}=\{\Psi_{\underline{n}}^{(L)},\,\, \underline{k}\in\Bbb{Z}^2 \}$. Once again, our main effort will be to find the coefficients
$c_{\underline{k}}^{(L)}$ in such a way that the following bi-orthogonalization requirement
$<\Psi_{\underline{n}}^{(L)},\varphi_{\underline{k}}^{(L)}>=\delta_{\underline{n}, \underline{k}}$ holds for all $\underline{n},\underline{k}$
in ${\Bbb Z}^2$. We also want $\F_\Psi^{(L)}$ to be a basis for $h_L$.

First of all, it is easy to check that $<\Psi_{\underline{n}}^{(L)},\varphi_{\underline{k}}^{(L)}>=\delta_{\underline{n}, \underline{k}}$,
$\forall\,\underline{n},\underline{k}\in {\Bbb Z}^2$, if and only if
$<\Psi_{\underline{n}}^{(L)},\varphi_{\underline{0}}>=\delta_{\underline{n}, \underline{0}}$, $\forall\,\underline{n}\in {\Bbb Z}^2$. Inserting
in this equality the expansion (\ref{32}) for $\Psi_{\underline{n}}^{(L)}$ we get $\sum_{\underline{k}\in {\Bbb
Z}^2}\overline{c_{\underline{k}}^{(L)}}\, I_{\underline{k}+\underline{n}}^{(L)}$. Multiplying both sides for
$e^{i\underline{p}\cdot\underline{n}}$ and summing up on $\underline{n}\in {\Bbb Z}^2$ we get \be
I^{(L)}(\underline{p})\,\overline{C^{(L)}(\underline{p})}=1, \qquad \mbox{a.e. in } \C, \label{34b} \en where $\C=[0,2\pi[\times[0,2\pi[$, and
where the functions above are defined as follows:
$$
C^{(L)}(\underline{p})=\sum_{\underline{k}\in {\Bbb Z}^2}\,c_{\underline{k}}^{(L)}\,e^{i\underline{p}\cdot\underline{k}}, \qquad
I^{(L)}(\underline{p})=\sum_{\underline{k}\in {\Bbb Z}^2}\,I_{\underline{k}}^{(L)}\,e^{i\underline{p}\cdot\underline{k}}.
$$
It is clear that the first formula can be inverted, producing $c_{\underline{k}}^{(L)}$ out of $C^{(L)}(\underline{p})$: \be
c_{\underline{k}}^{(L)}=\frac{1}{(2\pi)^2}\,\int_0^{2\pi}\int_0^{2\pi}\,\frac{e^{-i\underline{p}
\cdot\underline{k}}}{I^{(L)}(\underline{p})}\,d^2\underline{p},\label{34bis}\en at least when this integral exists. In deriving (\ref{34bis})
we have also used that $I^{(L)}(\underline p)$, and $C^{(L)}(\underline p)$ as a consequence, are real functions.

\vspace{2mm}

{\bf Remark:--} Taking into account the fact that $T_1$ and $T_2$ are unitary operators,
$\|\Psi_{\underline{n}}^{(L)}\|=\|\Psi_{\underline{0}}^{(L)}\|$ for all $\underline{n}$. Hence, all the elements in $\F_\Psi^{(L)}$ are well
defined vectors in $\Lc^2(\Bbb R)$ if and only if the coefficients $c_{\underline{k}}^{(L)}$ satisfy the following inequality:
$$
\|\Psi_{\underline{0}}^{(L)}\|^2=\sum_{\underline{k},\underline{l}\in\Bbb{Z}^2}\,\overline{c_{\underline{k}}^{(L)}}
\,c_{\underline{l}}^{(L)}\,I_{\underline{k}-\underline{l}}^{(L)}<\infty,
$$
where $I_{\underline{k}-\underline{l}}^{(L)}$ are deduced from (\ref{35}).

\vspace{2mm}

In analogy with what we have done in the previous section, we can prove the following results:

\begin{lemma} $I^{(L)}(\underline p)\in \Lc^2(\C)$ if and only if $\{I_{\underline{l}}^{(L)}\}\in l^2({\Bbb Z}^2)$,
and $\|I^{(L)}\|^2=(2\pi)^2\sum_{\underline{l}\in{\Bbb Z}^2}|I_{\underline{l}}^{(L)}|^2$. Also, $C^{(L)}(\underline p)\in \Lc^2(\C)$ if and
only if $\{c_{\underline{l}}^{(L)}\}\in l^2({\Bbb Z}^2)$, and $\|C^{(L)}\|^2=(2\pi)^2\sum_{\underline{l}\in{\Bbb
Z}^2}|c_{\underline{l}}^{(L)}|^2$. When they exist finite, both $C^{(L)}(\underline p)$ and $I^{(L)}(\underline p)$ are periodic and real
functions.
\end{lemma}

The proof of this Lemma is quite similar to that of Lemma \ref{lemma1} and will not be given here. We just want to remark that, using
(\ref{35}), it is clear that $\{I_{\underline{l}}^{(L)}\}\in l^2({\Bbb Z}^2)$ for all possible values of $L$. In fact, simple numerical
computations produce  $\sum_{\underline{l}\in{\Bbb Z}^2}|I_{\underline{l}}^{(L)}|^2=1.0883$, if $L=1$, $\sum_{\underline{l}\in{\Bbb
Z}^2}|I_{\underline{l}}^{(L)}|^2=1.00374$, if $L=2$, and so on. Of course, the larger the value of $L$, the closer the value of
$\sum_{\underline{l}\in{\Bbb Z}^2}|I_{\underline{l}}^{(L)}|^2$ to one.

\vspace{2mm}

{\bf Remark:--} Due to the decay behavior of the coefficients $I_{\underline{l}}^{(L)}$ it is clear that $I^{(L)}(\underline p )$ belongs to
other functional spaces. For instance, it is clear that it belongs to $\Lc^1(\C)$, as well as to $C(\C)$.

\vspace{2mm}

Proposition \ref{prop1} can be stated, in a slightly modified form, also in the present context. For that we first introduce the coefficients

\be d_{\underline{k}}^{(L)}=\frac{1}{(2\pi)^2}\,\int_0^{2\pi}\int_0^{2\pi}\,e^{-i\underline{p}
\cdot\underline{k}}\,I^{(L)}(\underline{p})\,d^2\underline{p}.\label{37}\en

Hence we have

\begin{prop} \label{prop2}
Let us assume that $\{d_{\underline{k}}^{(L)}\}, \{c_{\underline{k}}^{(L)}\}\in l^1({\Bbb Z}^2)$. Then $X^{(L)}$ and
$Y^{(L)}:=\sum_{\underline{k}\in\Bbb{Z}^2}\,d_{\underline{k}}^{(L)}\,T_1^{k_1}T_2^{k_2}$ are bounded operators, and $Y^{(L)}=(X^{(L)})^{-1}$.
Moreover, $\F_\Psi^{(L)}$ is the (only) basis for $h_L$ biorthogonal to $\F_\varphi^{(L)}$.
\end{prop}
The proof is very similar to that of Proposition \ref{prop1} and will not be repeated. Not surprisingly, it makes use of the following
summation rule: $\sum_{\underline{n}\in{\Bbb{Z}^2}}{c_{\underline{n}}^{(L)}}\,d_{\underline{l}-\underline{n}}^{(L)}=\delta_{\underline{l},
\underline{0}}$, which in particular implies that $\sum_{\underline{n}\in{\Bbb{Z}^2}}{c_{\underline{n}}^{(L)}}\,d_{-\underline{n}}^{(L)}=1$.

If we are under the assumptions of Proposition \ref{prop2}, and if $X^{(L)}$ is a positive operator, then the set
$\E^{(L)}=\{e^{(L)}_{\underline n}:=(X^{(L)})^{1/2}\varphi_{\underline n}^{(L)}, \,{\underline n}\in\Bbb{Z}^2\}$ is an orthonormal basis for
$h_L$ and $\F_\varphi^{(L)}$ and $\F_\Psi^{(L)}$ are biorthogonal Riesz bases.

\vspace{2mm}

{\bf Remark:--} A simple extension of Lemma \ref{lemma2} can be used to deduce that, at least if $L>1$, the sequences
$\{d_{\underline{k}}^{(L)}\}$, and $\{c_{\underline{k}}^{(L)}\}$ belong to $l^1({\Bbb Z}^2)$. The idea is that, rewriting
$I^{(L)}(\underline{p})=1+I^o_L(\underline{p})$, for all $ \underline{p}\in\C$, with
$$I^o_L(\underline{p})=\sum_{\underline{m}\in{\Bbb{Z}}^2\setminus{(0,0)}}\,(-1)^{L\,m_1\,m_2}\,e^{-\frac{\pi}{2}\,L(m_1^2
+m_2^2)}\,e^{i\underline{p}\cdot\underline{m}},$$ this last function can be easily estimated as follows:

$$ |I^o_L(P)| \leq \sum_{(m_1,m_2)\in\Bbb{Z}^2 \setminus{(0,0)}} e^{-\pi/2 L ({m_1}^2+{m_2}^2)}=  \left(\sum_{m\in\Bbb{Z}}e^{-\pi/2 L
{m}^2}\right)^2-1, $$ which is less that 1 for all $L=2,3,\ldots$. Indeed we find that $|I^o_2(P)| \leq 0.18$, $|I^o_3(P)| \leq 0.03$,
$|I^o_4(P)| \leq 0.007$, and so on. This implies that $I^{(L)}(\underline{p})$ and its inverse are well defined, square-integrable, functions
in $\Lc^2(\C)$, so that  Lemma \ref{lemma2} applies in its two-dimensional version. On the other hand, for $L=1$ we can only conclude that
$|I^o_1(P)| \leq 1.01$, which is not enough to get a similar conclusion. In fact, we will soon see that $L=1$ should be treated differently.

\subsection{Numerical results}

\vspace{4mm}

We show now how the coefficients $c_{\underline{n}}^{(L)}$ of the expansion (\ref{34bis}) can be computed perturbatively. For this we recall
that $I^{(L)}(\underline{p})=1+I^o_L(\underline{p})$. Then formula (\ref{34bis}) can be rewritten as follows: \be
c_{\underline{k}}^{(L)}=\frac{1}{(2\pi)^2}\,\int_0^{2\pi}\int_0^{2\pi}\,\frac{e^{-i\underline{p}
\cdot\underline{k}}}{{1+I^o_L(\underline{p})}}\,d\underline{p}= \frac{1}{(2\pi)^2}\,\int_0^{2\pi}\int_0^{2\pi}\,
e^{-i\underline{p}\cdot\underline{k}} \sum_{n=0}^\infty \left(-I^o_L(\underline{p})\right)^n\,d\underline{p}, \label{add1}\en at least if
$L\geq2$. Indeed, in this case, $|I^o_L(\underline{p})|$ is surely less than 1. Considering only the first two contributions of this expansion
we easily get \be c_{\underline{k}}\simeq \delta_{\underline{k},\,\underline{0}}-\left(1-\delta_{\underline{k},\,\underline{0}}\right)\,
(-1)^{Lk_1k_2}\,e^{-\frac{\pi}{2}L(k_1^2+k_2^2)}.\label{325}\en Of course, in order for this approximation to be meaningful, we further need to
restrict the sum in (\ref{34bis}) only to those $\underline k$ for which $\underline k=(\pm1,0), (0,\pm 1)$. In fact, a contribution like
$\underline k=(\pm 1,\pm1)$ could only be considered in (\ref{34bis}) if we also keep into account in (\ref{add1}) those contributions arising
from $I^o(\underline{P})^2$, which contain terms of the same order. For simplicity, and since our numerical results will show that ours is
already a very good approximation, all these contributions will be simply neglected here.  If we introduce the following subset of
${\Bbb{Z}}^2$, $\Gamma:=\{(1,0),(-1,0),(0,1),(0,-1)\}$, then we get the following expression for $\Psi_{\underline{n}}^{(L)}$: \be
\Psi_{\underline{n}}^{(L)}\simeq \varphi_{\underline{n}}^{(L)}-\,e^{-\frac{\pi}{2}L}\sum_{\underline{s}\in\Gamma}
\varphi_{\underline{n}+\underline{s}}^{(L)}. \label{326}\en

 It is
easy to check now that the set of the approximated vectors $\F_{\Psi}^{(L)}=\{\Psi_{\underline{n}}^{(L)}\}$ obtained in this way is
biorthogonal to $\F_{\varphi}^{(L)}$, with a very good approximation. Indeed if we compute the overlap between two neighboring vectors, for
instance between $\Psi_{1,0}^{(L)}$ and $\varphi_{0,0}^{(L)}$, we find that
$$
|<\Psi_{1,0}^{(L)}, \varphi_{0,0}^{(L)}>|=e^{(-3\pi/2) L} (2-e^{-\pi L})\simeq\left\{
\begin{array}{ll}
0.00016, \quad\mbox{if }L=2,    \\
0.000001, \quad\mbox{if }L=3.   \\
\end{array}
\right.$$

As for the normalization of the vectors, we find that $$|<\Psi_{0,0}^{(L)}, \varphi_{0,0}^{(L)}>|= 1-4e^{-\pi L}\simeq\left\{
\begin{array}{ll}
0.99253, \quad\mbox{if }L=2,  \\
0.99968, \quad\mbox{if }L=3,  \\
0.99999, \quad\mbox{if }L=4,  \\
\end{array}
\right.$$ and so on. We see that the approximation considered here works very well already for $L=2$\footnote{Needless to say,
$|<\Psi_{2,0}^{(L)}, \varphi_{0,0}^{(L)}>|$, $|<\Psi_{1,1}^{(L)}, \varphi_{0,0}^{(L)}>|$, $\ldots$,  are even smaller than $|<\Psi_{1,0}^{(L)},
\varphi_{0,0}^{(L)}>|$}. As expected, we get different conclusions for $L=1$. Indeed, if we try to repeat similar computations, we find that
$|<\Psi_{(1,0)}^{(1)}, \varphi_{(0,0)}^{(1)}>|\simeq0.018$ and that $|<\Psi_{(0,0)}^{(1)}, \varphi_{(0,0)}^{(1)}>|\simeq0.827$, which suggest
that the two functions are not biorthonormal. In other words, what appears to work very well for $L\geq2$, looks a dangerous procedure for
$L=1$. And there is more than this: as for the operator $X^{(L)}$ and $(X^{(L)})^{-1}$ we find that $X^{(L)}\simeq
\Id-e^{-\frac{\pi}{2}L}\sum_{\underline{s}\in\Gamma}\,T_1^{s_1}T_2^{s_2}= \Id-e^{-\frac{\pi}{2}L}K_L$ and $(X^{(L)})^{-1}\simeq
\Id+e^{-\frac{\pi}{2}L}\sum_{\underline{s}\in\Gamma}\,T_1^{s_1}T_2^{s_2}=\Id+e^{-\frac{\pi}{2}L}K_L$, where $K_L=T_1+T_1^{-1}+T_2+T_2^{-1}$.
 In order to check that $\Id+e^{-\frac{\pi}{2}L}K_L$  is a good approximation of the inverse of $X^{(L)}$ we observe that
$$\|X^{(L)}(X^{(L)})^{-1}-\Id\|=\|(X^{(L)})^{-1}X^{(L)}-\Id\|\leq 16 e^{-\pi L}=\left\{
\begin{array}{ll}
0.029879, \quad\mbox{if }L=2,  \\
0.001291, \quad\mbox{if }L=3,
\end{array}
\right.$$ and so on. It is interesting to notice that, for $L=1$, we get into serious troubles. Indeed our estimate appears rather poor:
$\|X^{(1)}(X^{(1)})^{-1}-\Id\|\leq 0.691423$. We are forced again to conclude that the case $L=1$ should be treated separately, and this is
exactly the content of the next section.

\subsection{ What if L=1?}

It is convenient to list few known facts on the $kq-$representation, \cite{zak,zak2}, which will be used in the following.

Let us introduce the following generalized functions: \be \Psi_{kq}(x)=\frac{1}{\sqrt{a}}\sum_{n\in\Bbb Z}e^{ikna}\delta(x-q-na), \label{51}\en
where $a=\sqrt{2\pi}$ (i.e., we are taking here $L=1$), and $k,q\in [0,a[$. It is known that $\Psi_{kq}(x)$ are eigenstates of $T_1$ and $T_2$,
$$
T_1\Psi_{kq}(x)=e^{iqa}\Psi_{kq}(x),\qquad T_2\Psi_{kq}(x)=e^{-ika}\Psi_{kq}(x),
$$
and that they {\em resolve the identity}: \be \int\int_{\Box}\overline{\Psi_{kq}(x)}\Psi_{kq}(x')dk\,dq=\delta(x-x'), \label{52}\en which we
can write, more schematically, as $\int\int_{\Box}|\Psi_{kq}\left>\right<\Psi_{kq}|dk,dq=\Id$. Here $\Box:=[0,a[\times[0,a[$.

Calling $\xi_x$ the generalized eigenstates of the position operator $\hat q$, $\hat q\xi_x=x\xi_x$, $x\in\Bbb R$, we know that an abstract
vector $f$ belongs to a certain Hilbert space $\hat\Hil$ if and only if $f(x):=\left<\xi_x,f\right>$ belongs to $\Lc^2(\Bbb R)$ or,
equivalently, if and only if $f(k,q):=\left<\Psi_{kq},f\right> = \frac{1}{\sqrt{a}}\sum_{n\in\Bbb Z}e^{ikna}f(q-na)$ belongs to $\Lc^2(\Box)$.

More interesting results on the $kq-$representation can be found in \cite{zak,zak2,jan}. A first, well known result, relating coherent states
and $kq-$representation concerns the completeness of $\F_\varphi^{(1)}$ in $\Lc^2(\Bbb R)$, \cite{bgz}: if $h(x)\in\Lc^2(\Bbb R)$ satisfies
$\left<h,\varphi_{\underline n}^{(1)}\right>_{\Lc^2}=0$ for all ${\underline n}\in{\Bbb Z}^2$, then $h(x)=0$ almost everywhere in $\Bbb R$. The
proof uses the resolution of the identity in (\ref{52}), together with the fact that the vector $\varphi_{\underline 0}$, in the
$kq-$representation, has just a single zero in $\Box$. In fact we have
$$
\varphi_{\underline 0}(k,q)=\left<\Psi_{kq},\varphi_{\underline 0}\right>=\sqrt{\frac{1}{\sqrt{2}\pi}}\,e^{-q^2/2}
\theta_3\left(\sqrt{\frac{\pi}{2}}(k-iq),e^{-\pi}\right),
$$
where $\theta_3$ is an elliptic theta function, \cite{grad}. The point $P_0=(k_0,q_0)=\left(\sqrt{\frac{\pi}{2}},\sqrt{\frac{\pi}{2}}\right)$
is the only zero for $\varphi_{\underline 0}(k,q)$ in $\Box$: $\varphi_{\underline 0}(k_0,q_0)=0$. We refer to \cite{bgz} for the details on
this proof.

The relevant application of $kq$-representation for us consists  in finding a vector $\Psi_{\underline 0}^{(1)}$ producing first of all, as in
(\ref{32}), a set of vectors $\F_\Psi^{(1)}$ in $\Lc^2(\Bbb R)$ biorthogonal to $\F_\varphi^{(1)}$. To achieve this aim, let us call
$\Psi_{\underline 0}^{(1)}(k,q)=\left<\Psi_{kq},\Psi_{\underline 0}^{(1)}\right>$ our unknown function in the $k,q$ variables. Then we have
$$
\left<\Psi_{\underline n}^{(1)},\Psi_{kq}\right>=\left<\Psi_{\underline 0}^{(1)},T_2^{-n_2}T_1^{-n_1}\,\Psi_{kq}\right>=
e^{ikan_2-iqan_1}\overline{\Psi_{\underline 0}^{(1)}(k,q)}.
$$
Using now the resolution of the identity (\ref{52}) we get
$$
\delta_{{\underline n},{\underline 0}}=\left<\Psi_{\underline n}^{(1)},\varphi_{\underline 0}\right>=\int\int_{\Box} \left<\Psi_{\underline
n}^{(1)},\Psi_{kq}\right>\left<\Psi_{kq},\varphi_{\underline 0}\right>\,dk\,dq=$$ \be=\int\int_{\Box}
e^{ikan_2-iqan_1}\overline{\Psi_{\underline 0}^{(1)}(k,q)}\varphi_{\underline 0}(k,q)\,dk\,dq. \label{extra1}\en The formal solution of this
equation is easily found: $\Psi_{\underline 0}^{(1)}(k,q)=\frac{1}{2\pi\overline{\varphi_{\underline 0}(k,q)}}$, which has a single singularity
in $P_0$. Incidentally we deduce that, as for $\varphi_{\underline 0}$, also $\Psi_{\underline 0}^{(1)}$ appears to be independent of $L$.
 The solution in the coordinate representation is: \be \Psi_{\underline
0}(x)=\frac{1}{2\pi}\int\int_{\Box}\Psi_{kq}(x)\frac{1}{\overline{\varphi_{\underline 0}(k,q)}}\,dk\,dq, \label{53}\en where we have
removed the unessential suffix $(1)$. The (momentary)
conclusion seems therefore that, using this $\Psi_{\underline 0}(x)$, the set $\F_\Psi^{(1)}$ could be constructed. However, a simple argument shows
that this possibility is not allowed. The reason is the following:

first, we know that $b\varphi_{\underline n}^{(1)}=z_{\underline n}\varphi_{\underline n}^{(1)}$, for all ${\underline n}\in{\Bbb Z}^2$. Hence,
if the biorthogonal basis $\F_\Psi^{(1)}$ can be defined, we would have
$$
\left<\Psi_{\underline n}^{(1)},b\,\varphi_{\underline m}^{(1)}\right>=z_{\underline m}\delta_{{\underline n},{\underline m}}=
\left<\overline{z_{\underline n}}\,\Psi_{\underline n}^{(1)},\varphi_{\underline m}^{(1)}\right>,
$$
as well as $\left<\Psi_{\underline n}^{(1)},b\,\varphi_{\underline m}^{(1)}\right>=\left<b^\dagger\Psi_{\underline n}^{(1)},\varphi_{\underline
m}^{(1)}\right>$. Hence, since $\F_\varphi^{(1)}$ is complete in $\Lc^2(\Bbb R)$, we must have $b^\dagger\Psi_{\underline
n}^{(1)}=\overline{z_{\underline n}}\Psi_{\underline n}^{(1)}$, which in particular means that $\Psi_{\underline 0}$ should be annihilated by
$b^\dagger$: $b^\dagger\Psi_{\underline 0}=0$. But this equation, in the representation space, has a single solution which is not
square-integrable: $\Psi_{\underline 0}(x)=Ne^{x^2/2}$. This implies that the function $\Psi_{\underline 0}(x)$ in (\ref{53}) cannot belong to
$\Lc^2(\Bbb R)$ either.  The conclusion is therefore that, as in \cite{bt1}, the best we can do is to find biorthogonal sets in suitable
subspaces of $\Lc^2(\Bbb R)$, but not in all of this space.

\subsection{What changes when $L=2$}

Interestingly enough, the previous reasoning does not apply for $L>1.$ Indeed in this case  $\F_\varphi^{(L)}$ is not complete in $\Lc^2(\Bbb R)$.
Therefore the equality $\left<\overline{z_{\underline n}}\,\Psi_{\underline n}^{(L)},\varphi_{\underline
m}^{(L)}\right>=\left<b^\dagger\,\Psi_{\underline n}^{(L)},\varphi_{\underline m}^{(L)}\right>$, $\forall {\underline m}\in {\Bbb Z}^2$, does
not necessarily imply that $b^\dagger\,\Psi_{\underline n}^{(L)}=\overline{z_{\underline n}}\,\Psi_{\underline n}^{(L)}$.

This lack of completeness is also reflected by the fact that the set $\{e^{ika n_2-iqan_1}, {\underline n}\in {\Bbb Z}^2\}$ is not complete in
$\Lc^2(\Box)$, if $a^2=4\pi$ (i.e., if $L=2$). For this reason, the same computations giving rise to (\ref{extra1}), produce now the following
equality:
$$
2\int_{0}^{2\pi/a}dk \int_0^{2\pi/a} e^{ika n_2-iqan_1} \left[\overline{\Psi_{\underline 0}(k,q)}\varphi_{\underline 0}(k,q)
+ \overline{\Psi_{\underline 0}(k,q+\frac{2\pi}{a})}\varphi_{\underline 0}(k,q+\frac{2\pi}{a})\right]dq=\delta_{\underline{n},\underline{0}},
$$
which implies that \be{\varphi_{\underline 0}(k,q)} \overline{\Psi_{\underline 0}(k,q)}+ {\varphi_{\underline 0}(k,q+\frac{2 \pi}{a})}
\overline{\Psi_{\underline 0} (k,q+\frac{2 \pi}{a})}=1/2,\label{extra2}\en a.e. for $(k,q)\in
\left[0,\frac{2\pi}{a}\right[\times\left[0,\frac{2\pi}{a}\right[$. This equation may have solutions which are different from the one deduced
for $L=1$, and this explains why for $L=2$ (as well as for $L=3, 4, \ldots$) our previous conclusion about the non existence of the set
$\F_\Psi^{(1)}$ cannot be extended.

Summarizing, we have discussed so far two possible strategies to construct a set $\F_\Psi^{(L)}$ out of the given set of coherent states
$\F_\varphi^{(L)}$. These two possibilities, which work if $L=2,3,4,\ldots$ but not for $L=1$, are the following:

$\bullet$ a perturbative expansion as in (\ref{31b}) and (\ref{32}). This works directly in the coordinate space, and produces a function
$\Psi_{\underline 0}(x)$ and, from this square integrable function, the set $\F_\Psi^{(L)}$ we were looking for.

$\bullet$ {\em the $(k,q)$-way to biortogonality}: this is more delicate, but, in principle, non perturbative: one has to find the solution for
(\ref{extra2}), or for the extended version of this equation for $L=3, 4,\ldots$, and then use the resolution of the identity to go back to the
space $\Lc^2(\Bbb R)$.

\vspace{2mm}

{\bf Remark:--} it is very easy to check that some of the features of $\F_\varphi^{(L)}$ are shared by $\F_\Psi^{(L)}$. First of all, by
construction, the set $\F_\Psi^{(L)}$ is stable under the action of $T_1$ and $T_2$. Moreover, if we introduce the operator
$B^{(L)}:=X^{(L)}b(X^{(L)})^{-1}$, it is easy to check that $\Psi_{\underline{n}}^{(L)}$ is an eigenstate of $B^{(L)}$, with eigenvalue
$z_{\underline{n}}$:
$$
B^{(L)}\Psi_{\underline{n}}^{(L)}=\left(X^{(L)}b(X^{(L)})^{-1}\right)\left(X^{(L)}\varphi_{\underline{n}}^{(L)}\right)=X^{(L)}b
\varphi_{\underline{n}}^{(L)}=z_{\underline{n}}X^{(L)}\varphi_{\underline{n}}^{(L)}=z_{\underline{n}} \Psi_{\underline{n}}^{(L)}.
$$
It is also possible, in principle, to introduce an extended version of the Heisenberg uncertainty relation. However, rather than going in this
direction, in the next section we will briefly discuss the relation of our results with pseudo-hermitian quantum mechanics.

\section{$N\geq 3$ and relations with Pseudo-hermitian quantum mechanics}

The general procedure introduced in Section II, and adopted in Section III for coherent states, can be easily extended to all possible $N$, at least when $A_1$, $A_2$, $\ldots$, $A_N$ mutually commute. The idea is exactly the same: starting with $\varphi_{k_1,\ldots,k_N}=A_1^{k_1}\ldots A_N^{k_N}\varphi$, and assuming that they are linearly independent for all $k_j\in\Bbb Z$, we look for a new vector $\Psi$ as $\Psi=\sum_{k_1,\ldots,k_N\in \Bbb Z}c_{\bf k}\varphi_{\bf k}$. Here, to simplify the notation, ${\bf k}=(k_1,\ldots,k_N)$. Then, as in Section II, we can write
$$
\Psi_{\bf n}=A_1^{n_1}\ldots A_N^{n_N}\Psi=X\varphi_{\bf n},\quad \mbox{ where }\quad X=\sum_{\bf k\in {\Bbb Z}^N}c_{\bf k} A_1^{n_1+k_1}\ldots A_N^{n_N+k_N}.
$$
Due to the commutativity of the $A_j$'s, it is easy to deduce that $\left<\Psi_{\bf n},\varphi_{\bf k}\right>=\delta_{\bf n,k}$ if and only if $\left<\Psi_{\bf n},\varphi_{\bf 0}\right>=\delta_{\bf n,0}$. Now, multiplying both sides of this last equation by $e^{i{\bf p}\cdot {\bf n}}$, with ${\bf p}=(p_1,\ldots,p_n)$ and $p_j\in\left]0,2\pi\right]$, and summing on $\bf n$, we recover equation (\ref{II8}), with essentially the same definitions as in (\ref{II9}). The next steps as in Section II can be carried out. Of course, results on convergence of the several series involved in our computations strongly depend on the explicit form of the $A_j$, as it appears clear in Section III. These are the extra information needed to make our results completely rigorous, as in the previous sections.

\vspace{3mm}

An interesting aspect of our problem is the following: it is well known that the eigenvectors of non self-adjoint operators are not mutually orthogonal. However, particularly in connections with some physical systems, \cite{bender,ali}, they could produce biorthogonal sets and, sometimes, biorthogonal bases. Also, these bases could be related to some family of raising and lowering operators, \cite{bagmore}, and to some non self-adjoint number-like operator, of the same kind one often find in the literature on extended harmonic oscillators, \cite{eho}. For this reason, it is interesting to show here how the families $\F_\varphi^{(L)}$, $\F_\Psi^{(L)}$ and $\E^{(L)}$ give rise to three different, isospectral,
operators, related by some intertwining relations. In particular, two of these operators will turn out to be one the adjoint of the other, while the third
operator is self-adjoint.

Let us first recall that, for us, $\F_\varphi^{(L)}$ and $\F_\Psi^{(L)}$ are biorthogonal bases in $h_L$. This means, using the Dirac
notation\footnote{This notation will be adopted in all this section.}, that the following resolutions of the identity hold:
$\Id_L=\sum_{\underline{n}}|\varphi_{\underline{n}}^{(L)}\left>\right<\Psi_{\underline{n}}^{(L)}|=
\sum_{\underline{n}}|\Psi_{\underline{n}}^{(L)}\left>\right<\varphi_{\underline{n}}^{(L)}|$, where $\Id_L$ is the identity operator in $h_L$.
Under the assumptions of Proposition \ref{prop2}, $X^{(L)}$ and $(X^{(L)})^{-1}$ are bounded operators. If we define
$S_\Psi^L=\sum_{\underline{n}}|\Psi_{\underline{n}}^{(L)}\left>\right<\Psi_{\underline{n}}^{(L)}|$ and
$S_\varphi^L=\sum_{\underline{n}}|\varphi_{\underline{n}}^{(L)}\left>\right<\varphi_{\underline{n}}^{(L)}|$, it is easy to check, first of all,
that $S_\Psi^L(X^{(L)})^{-1}=X^{(L)}S_\varphi^L$.

Moreover, if $X^{(L)}$ is positive, then $\E^{(L)}=\{e_{\underline{n}}^{(L)}:=(X^{(L)})^{1/2}\varphi_{\underline{n}}^{(L)},\,\underline{n}\in
{\Bbb Z}^2 \}$ is an o.n. basis for $h_L$. Then we have, for $f\in h_L$,
$$
S_\varphi^L\,f=\sum_{\underline{n}}\left<\varphi_{\underline{n}}^{(L)},f\right>\varphi_{\underline{n}}^{(L)}=
\sum_{\underline{n}}\left<e_{\underline{n}}^{(L)},(X^{(L)})^{-1/2}f\right>\,(X^{(L)})^{-1/2}e_{\underline{n}}^{(L)}=(X^{(L)})^{-1}f.
$$
Hence $S_\varphi^L=(X^{(L)})^{-1}$. Analogously one finds that $S_\Psi^L=X^{(L)}$, and the relation $S_\Psi^L(X^{(L)})^{-1}=X^{(L)}S_\varphi^L$
is clearly verified.

Let us now introduce (formally\footnote{These operators could be unbounded, so that a rigorous definition implies knowledge of their domains. This aspect is not very interesting in our analysis, so that will be neglected.}) the following operators: $h=\sum_{\underline{n}}\epsilon_{\underline{n}}|e_{\underline{n}}^{(L)}\left>\right<e_{\underline{n}}^{(L)}|$,
$H=\sum_{\underline{n}}\epsilon_{\underline{n}}|\varphi_{\underline{n}}^{(L)}\left>\right<\Psi_{\underline{n}}^{(L)}|$, and
$H^\dagger=\sum_{\underline{n}}\epsilon_{\underline{n}}|\Psi_{\underline{n}}^{(L)}\left>\right<\varphi_{\underline{n}}^{(L)}|$. Here
$\{\epsilon_{\underline{n}}\}$ is an arbitrary sequence of real numbers (bounded or not, not necessarily positive, for our purposes). It is
clear that $h=h^\dagger$. These three operators have, by construction, the same eigenvalues but different eigenvectors. Indeed we have:
$$
he_{\underline{n}}^{(L)}=\epsilon_{\underline{n}}e_{\underline{n}}^{(L)}, \quad
H\varphi_{\underline{n}}^{(L)}=\epsilon_{\underline{n}}\varphi_{\underline{n}}^{(L)}, \quad
H^\dagger\Psi_{\underline{n}}^{(L)}=\epsilon_{\underline{n}}\Psi_{\underline{n}}^{(L)}.
$$
Moreover, among the others, the following intertwining relations can be deduced:
$$
X^{(L)}H= H^\dagger X^{(L)}, \quad h\left(S_\varphi^L\right)^{1/2}=\left(S_\varphi^L\right)^{1/2} H^\dagger, \quad
H\left(S_\varphi^L\right)^{1/2}=\left(S_\varphi^L\right)^{1/2}h.
$$
Hence we are back to the general structure considered in several papers, see \cite{bender}, \cite{ali} and references therein, concerning
quantum mechanics with non self-adjoint hamiltonians. There are also obvious connections with the theory of intertwining operators, used in the construction of exactly solvable
models, see for instance \cite{intop}. Incidentally we observe that the one described here is a general scheme which can be proposed starting with any pair of biorthogonal bases. What
is more related to the construction proposed in Section III is the existence of two lowering operators ($b$ and $B^{(L)}$) and the fact that the intertwining
operators $S_\varphi^L$ and $S_\Psi^L$ can be written in terms of the operator $X^{(L)}$ in (\ref{33b}).


\section{Conclusions}

We have shown how and under which conditions biorthogonal sets of  {\em coherent-like} vectors can be constructed, closed under the action of certain unitary operators, and how they produce (non) self-adjoint hamiltonians and intertwining  operators. We have seen that the main mathematical problem in our analysis consists in the analysis of the convergence of some series, and this is related to the analytic expression of the operators defining the biorthogonal sets.

A possible application of our general procedure is surely  the construction of some quantization procedure, see \cite{gazbook} and references therein. In the usual literature on this subject one uses standard coherent states, which produces a resolution of the identity. Our biorthogonal sets also produce a weak resolution, on some suitable subspace of the Hilbert space, so that we expect that interesting results can be deduced. Also some more mathematical aspects of the construction, mainly related to the unboundedness of the lowering and raising operators introduced in the previous section are presently object of investigation, also in connection with their algebraic properties.

\section*{Acknowledgments}

This work has been financially supported in part by GNFM and in part by MURST.


\begin{thebibliography}{99}

\bibitem{dau} I. Daubechies, {\em Ten Lectures on Wavelets},
           Society for Industrial and Applied Mathematics, Philadelphia, (1992)
\bibitem{bt1} F. Bagarello, S. Triolo, {\em Invariant analytic
orthonormalization procedure with an application to coherent states},  J. Math. Phys., {\bf 48}, 043505 (2007)
\bibitem{bt2}  F. Bagarello, M.R. Abdollahpour, S. Triolo, {\em An invariant
analytic orthonormalization procedure with applications}, J. Math. Phys., {\bf 48}, 103513 (2007)
\bibitem{bender} C. Bender, {\em Making Sense of Non-Hermitian Hamiltonians}, Rep. Progr.  Phys., {\bf 70},  947-1018 (2007)

\bibitem{ali} A. Mostafazadeh, {\em Pseudo-Hermitian representation of Quantum Mechanics}, Int. J. Geom. Methods Mod. Phys. {\bf 7}, 1191-1306, (2010)


\bibitem{bases}  C. Heil, {\em A basis theory primer: expanded edition}, Springer, New York, (2010).
 Christensen O., {\em An Introduction to Frames and Riesz Bases}, Birkh\"auser, Boston, (2003)







\bibitem{aag} S. T. Ali, J.-P. Antoine and J.-P. Gazeau, {\em Coherent States, Wavelets and
their Generalizations\/}, Springer-Verlag, New York (2000)
\bibitem{gazbook} J. P.Gazeau,  {\em Coherent States in Quantum Physics},  Wiley-VCH, Berlin (2009)
\bibitem{zak} J. Zak, {\em Dynamics of Electrons in Solids in External Fields},
           Phys. Rev., {\bf 168}, 686, (1968)
\bibitem{zak2} J. Zak, {\em The kq-representation in the dynamics of electrons in solids}, Solid State Physics, H. Ehrenreich, F. Seitz, D. Turnbull Eds., Academic, New York (1972), Vol. 27
\bibitem{jan} Janssen A.J.E.M.,  {\em Bargmann transform, Zak transform, and Coherent states},
           J. Math. Phys., {\bf 23}, 720, (1982)
\bibitem{bgz} H. Bacry, A. Grossmann, J. Zak, {\em Proof of Completeness of Lattice States in the kq Representation},
           Phys. Rev. B, {\bf 12}, 1118, (1975)
\bibitem{grad} I. S. Gradshteyn and I. M. Ryzhik, {\em Table
          of Integrals, Series and Products}, Academic Press,
        New York and London,  (1980)
\bibitem{bagmore} F. Bagarello, {\em More mathematics for pseudo-bosons},  J. Math. Phys., {\bf 54}, 063512 (2013)
\bibitem{eho} C. M. Bender, H. F. Jones, {\em Interactions of Hermitian and non-Hermitian Hamiltonians}, J. Phys. A, {\bf 41}, 244006  (2008); Jun-Qing Li, Qian Li, Yan-Gang Miao, {\em Investigation of PT-symmetric Hamiltonian Systems from an Alternative Point of View}, Commun. Theor. Phys., {\bf 58}, 497; F. Bagarello, {\em From self to non self-adjoint harmonic oscillators: physical consequences and mathematical pitfalls},
Phys. Rev. A, DOI: 10.1103/PhysRevA.88.032120; J. da Provid$\hat e$ncia, N. Bebiano, J.P. da Provid$\hat e$ncia, {\em Non hermitian operators with real spectrum in quantum mechanics}, ELA, {\bf 21}, 98-109 (2010).

\bibitem{intop} Kuru S., Tegmen A., Vercin A., {\em Intertwined isospectral potentials in an arbitrary dimension},
J. Math. Phys, {\bf 42}, No. 8, 3344-3360, (2001); Kuru S., Demircioglu B., Onder M., Vercin A., {\em Two families of superintegrable and
isospectral potentials in two dimensions}, J. Math. Phys, {\bf 43}, No. 5, 2133-2150, (2002); Samani K. A., Zarei M., {\em Intertwined
hamiltonians in two-dimensional curved spaces}, Ann. of Phys., {\bf 316}, 466-482, (2005); Bagarello, {\em Non isospectral hamiltonians, intertwining operators and hidden hermiticity}, Phys. Lett. A, {\bf 376}, 70-74 (2011).







\end{thebibliography}
\end{document}